\documentclass[runningheads]{llncs}
\usepackage[T1]{fontenc}

\usepackage[letterpaper, margin=1in]{geometry}

\usepackage[normalem]{ulem}  

\usepackage{thm-restate}
\usepackage{amsmath,amssymb}
\usepackage{graphicx}

\usepackage[numbers,sort&compress]{natbib}
\usepackage{xifthen}
\usepackage{multirow}
\usepackage{array}

\usepackage{xcolor}

\usepackage{thmtools} 
\usepackage{thm-restate}
\usepackage[ruled,linesnumbered,vlined]{algorithm2e}
\usepackage{cleveref} 

\usepackage{float}
\usepackage{xspace}
\usepackage{subfigure}
\usepackage{bbm}
\usepackage{natbib}

\newcommand{\IC}{\textnormal{IC}\xspace}
\newcommand{\LT}{\textnormal{LT}\xspace}

\newcommand{\InfMax}{\textsc{InfMax}\xspace}

\renewcommand{\Pr}[2][]{ \ifthenelse{\isempty{#1}}
  {\mathbf{Pr}\left[#2\right]} {\mathbf{Pr}_{#1}\left[#2\right]} }
\newcommand{\E}[2][]{ \ifthenelse{\isempty{#1}}
  {\mathbf{E}\left[#2\right]}
  {\mathbf{E}_{#1}\left[#2\right]} }
\begin{document}
\title{Algorithms and Complexity of Influence Maximization on Directed Acyclic Graphs}
\titlerunning{Algorithms and Complexity of Influence Maximization on DAGs}
%
\author{Panfeng Liu\inst{}\orcidID{0009-0000-5996-7206} \and
Biaoshuai Tao\inst{}\orcidID{0000-0003-4098-844X}}
\authorrunning{P. Liu and B. Tao}
%
\institute{Shanghai Jiao Tong University, Shanghai, China\\
\email{\{liupf22,bstao\}@sjtu.edu.cn}}
\maketitle              
\begin{abstract}

This paper investigates the influence maximization problem under the Independent Cascade (IC) and Linear Threshold (LT) models. While this problem is known to be APX-hard on general graphs, we explore its computational limits by focusing on Directed Acyclic Graphs (DAGs) and more restricted tree structures. Our primary result demonstrates that influence maximization remains APX-hard on DAGs under the LT model, suggesting that the absence of cycles is insufficient to achieve a polynomial-time approximation scheme (PTAS). In contrast, we show that the problem becomes tractable when the topology is further restricted to out-arborescences and in-arborescences. Specifically, for out-arborescences, we show that the IC model and the LT model are equivalent, and we develop exact polynomial-time algorithms based on dynamic programming that leverage the unique path properties of these structures.
For in-arborescences, it is known that the problem is polynomial-time solvable under the LT model, and it is NP-hard under the IC model.
We complement these results by presenting a fully polynomial-time approximation scheme (FPTAS) for the IC model.

\end{abstract}
\section{Introduction}\label{sec:Intro}

The \emph{influence maximization} (\InfMax) problem, which is defined to \emph{seed} a set of nodes in the social network to maximize the spreading of the resulting \emph{cascade} \citep{KKT15,DPRM2001,BJJR1987,RD2002,Rigobon2002,PV01,KR2010,chen2009approximability,angell2017don,schoenebeck2019beyond,schoenebeck2022think}, where a \emph{cascade} is a basic social network process that captures the character of the information diffusion in a social network, in which a number of initially infected seeds process a certain attribute of information and spread this attribute to their neighbor.

Two most well-known cascade models are the independent cascade (\IC) model and the linear threshold (\LT) model, both introduced in the seminal work of Kempe et al.~\cite{KKT15}. These two models are studied almost exclusively in the past literature of this field, and we will focus on these two models in this paper.  
In the \IC model, a newly activated node $u$ attempts to activate each of its inactive neighbors $v$ with a fixed probability, independently of other infections (we use the words ``activation'' and ``infection'' interchangeably throughout this work).
In the \LT model, if the graph is unweighted, each non-seed vertex is assigned a threshold randomly and independently from the interval $[0,1]$. A vertex becomes active when the fraction of its infected neighbors exceeds its threshold.
Formal definitions of the two models are available in Sect.~\ref{sec:Preliminary}.

The computational complexity and approximability of \InfMax are well-understood in general graphs. For both the \IC and \LT models, the problem is known to be APX-hard on directed and undirected graphs~\cite{KKT15,schoenebeck2020influence}. On the positive side, since the influence function $\sigma(\cdot)$ is submodular under these models~\citep{MR2010}, a simple greedy algorithm achieves a $(1-1/e)$-approximation ratio~\citep{nemhauser1978analysis, KKT15}. For undirected networks, this approximation guarantee can be slightly improved~\cite{khanna2014influence, schoenebeck2020limitations}. However, for directed graphs under the \IC model, the $(1-1/e)$ bound remains tight unless $\text{P} = \text{NP}$~\citep{KKT15}. These hardness results suggest that achieving exact solutions or better approximation ratios requires restricting the graph topology.

In this paper, we focus on Directed Acyclic Graphs (DAGs) and their variants. DAGs represent a natural and significant special case for influence propagation, as influence in the real world often follows a clear hierarchical and unidirectional structure. For instance, in social media or professional networks, such as citation networks, information cascades, influence typically flows from high-profile figures or domain experts to their followers in a top-down, non-recurrent manner. 
Moreover, DAG-based decompositions are widely used in the analysis of diffusion processes on general graphs. DAGs are among the most fundamental graph classes in algorithm design and complexity theory. 
The existence of topological orderings often makes problems more tractable on DAGs, providing insight into the boundary between tractable and intractable cases. 
Despite the absence of cycles, we first present our primary hardness result: \InfMax remains APX-hard on DAGs under the \LT model, indicating that the problem's inherent difficulty persists even in such structured settings. Motivated by this boundary, we further investigate two more restricted structures: out-arborescences and in-arborescences. We demonstrate that the unique path properties of these trees allow for polynomial-time algorithms or approximation schemes via dynamic programming, thereby providing a complete characterization of the transition from tractability to APX-hardness in hierarchical networks.

\subsection{Our results}
We first study the approximability of \InfMax on directed acyclic graphs.
For the \IC model, the same proof in Kempe et al.~\cite{KKT15} can show that \InfMax remains NP-hard to approximate to within a factor of $(1-1/e+\varepsilon)$ (for any $\varepsilon>0$), as the reduction used by Kempe et al. uses a directed acyclic graph.
In this paper, we show that \InfMax is APX-hard for the \LT model for directed acyclic graphs.
Specifically, there exists a universal constant $\tau>0$ such that \InfMax is NP-hard to approximate to within a factor of $(1-\tau)$.
See Sect.~\ref{sec:apx-hardness} for details.

Motivated by the hardness of directed acyclic graphs, we then focus on directed trees.
For out-arborescences, in Sect.~\ref{subsec:out-arborescence}, we show that the unique path property permits a polynomial-time dynamic programming algorithm. Crucially, we prove that the \LT and \IC models are equivalent on this structure, rendering our dynamic programming approach universally applicable to both propagation models in this setting.
For in-arborescences, the \IC and \LT models are no longer equivalent.
Wang et al.~\cite{wang2016bharathi} show that \InfMax under \LT model admits a polynomial-time algorithm, while Lu et al.~\cite{lu2017solution} show that \InfMax under the \IC model is still NP-hard.
We complement these results by presenting a fully polynomial-time approximation scheme for the \IC model in Sect.~\ref{subsec:in-arborescence}.

\section{Preliminaries}\label{sec:Preliminary}
Throughout this paper, a social network is represented by a weighted directed graph $G=(V, E, P)$, where $V$ is the node (i.e., user) set with $|V|=n$, $E$ is the edge (i.e., the connections between users) set with $|E|=m$ and $P=\{p_e\in(0,1]\}_{e \in E}$ is the edge weights.
For a directed edge $(u,v) \in E$, we say $u$ (resp. $v$) is the incoming (resp. outgoing) neighbor of $v$ (resp. $u$), and $(u,v)$ is the outgoing (resp. incoming) edge of $u$ (resp. $v$). 

This paper focuses on the class of \emph{directed acyclic graphs} (DAGs) and two special sub-classes: \emph{in-arborescences} and \emph{out-arborescences}.

\begin{definition}
Let $G=(V,E)$ be a directed acyclic graph (DAG). 
For a vertex $v \in V$, an in-arborescence rooted at $v$ is a directed tree in which every vertex has a unique directed path to $v$, 
i.e., all edges are oriented toward the root $v$. 
An out-arborescence rooted at $v$ is a directed tree in which there exists a unique directed path from $v$ to every other vertex, i.e., all edges are oriented away from the root.
\end{definition}

\paragraph{Diffusion Models.} 
A \emph{diffusion model} $\Gamma$ is a (possibly random) function that maps from a seed set $S$ (the vertices that are initially infected) to a vertex set $\Gamma(S)$ (the set of influenced vertices at the end of the spreading).

\begin{definition}[Independent Cascade Model (\IC) \cite{KKT15}]\label{def:IC}
         Given a social network (directed graph) $G=(V, E, P)$, the \IC model assigns the state of each node as \emph{active} or \emph{inactive}.
         On the input seed set $S \subseteq V$, \IC outputs the set $\Gamma_\IC(S)$ as follows:
         \begin{enumerate}
             \item At timestamp $0$, only nodes in $S$ are active.
             \item At each timestamp $t=1,2,\dots$, each newly activated node $u$ from the previous timestamp gets one chance to activate its inactive outgoing neighbor $v$ with probability $p_{(u,v)}$. The attempts to activate neighbors are independent of each other. If multiple incoming neighbors of an inactive node attempt to activate it, each attempt is considered separately with its own probability.
             \item The diffusion process terminates when no additional activation occurs, and \IC outputs $\Gamma_\IC(S)$ as the set of active nodes.
         \end{enumerate}
    \end{definition}

\begin{definition}[Linear Threshold Model (\LT) \cite{KKT15}]\label{def:LT}
         Given a social network (directed graph) $G=(V, E, P)$ where the edge weights satisfy  $\sum_{u \in N(v)}p_{(u,v)} \leq 1$ for each $v \in V$, 
         the \LT model assigns the state of each node as \emph{active} or \emph{inactive}.
        Given the input seed set $S \subseteq V$, \LT outputs the set $\Gamma_{\LT}(S)$ as follows:
         \begin{enumerate}
            \item At the beginning, for each vertex $v$, a threshold $\theta_v \in [0,1]$ is sampled uniformly at random independently.
             \item At timestamp $0$, only nodes in $S$ are active.
             \item At each timestamp $t=1,2,\dots$, a node $v$ is active if the sum of the weights of
edges from $v$’s active in-neighbors to $v$ exceeds the threshold $\theta_v$:$$\sum_{u:u \in N(v)\text{ and }u\text{ is active}}p_{(u,v)} \geq \theta_v$$ 
             \item The diffusion process terminates when no additional activation occurs, and \LT outputs $\Gamma_{\LT}(S)$ as the set of active nodes. 
         \end{enumerate}
    \end{definition}

\begin{definition}[Influence Maximization (\InfMax) \cite{KKT15}]
         Given a directed weighted graph $G=(V,E,P)$, a positive integer $k$ and a diffusion model $\Gamma$, the \InfMax problem aims to find a seed set $S \subseteq V$ with $|S|\leq k$ that maximizes the expected spread $\sigma_\Gamma(S)$  under the given diffusion model. 
\end{definition}

Given a diffusion model $\Gamma$, let $\sigma_\Gamma(S)=\mathbb{E}(|\Gamma(S)|)$ be the expected number of infected vertices.
The goal is to find $S$ that maximizes $\sigma_\Gamma(S)$.
It is known that computing $\sigma_\Gamma(\cdot)$ is $\#$P-hard for both \IC and \LT models~\cite{CYZ10,CCY10,KKT15}. On the other hand, by simple Monte-Carlo samplings, we can easily obtain a fully polynomial-time randomized approximation scheme (FPRAS) for computing $\sigma_\Gamma(\cdot)$. 
In addition, for directed acyclic graphs, the evaluation of $\sigma_\Gamma(\cdot)$ under the \LT model becomes tractable: since the influence propagation is acyclic, given a seed set, we can compute the probability of infection for each vertex one by one following the topological order of the graph~\cite{CYZ10}.
This problem remains $\#$P-hard for the \IC model even for directed acyclic graphs~\cite{nguyen2013budgeted}.
However, as we will see later, the APX-hardness of the influence maximization problem with the \IC model on directed acyclic graphs does not come from the $\#$P-hardness of evaluating $\sigma_\Gamma(\cdot)$.
For in-arborescences and out-arborescences, computing $\sigma_\Gamma(\cdot)$ for the \IC model becomes tractable: we can compute the probability of infection for each vertex in the bottom-up order (for in-arborescences) or the top-down order (for out-arborescences).

In this paper, we omit the subscript \(\Gamma\) of $\sigma_\Gamma$ when it is clear from the context.

\section{APX-Hardness of \InfMax on Directed Acyclic Graphs}
\label{sec:apx-hardness}

In the context of propagation modeling, directed acyclic graphs eliminate cyclic feedback, ensuring that the diffusion of influence follows a clear temporal or hierarchical order, which seemingly significantly simplifies the propagation dynamics. 
However, this section demonstrates that such structural simplification is insufficient to significantly reduce the computational hardness of the influence maximization problem. We will prove that even when the propagation graph is restricted to a DAG, the influence maximization problem remains difficult in the sense of approximation; specifically, the problem is APX-hard.

For the \IC model, there is a straightforward reduction from the \emph{max-k-coverage} problem (selecting $k$ subsets of a set with the largest union) to the \InfMax problem. Firstly, create a vertex for each subset, and a large number $w$ of vertices for each element. Then, create a directed edge from a vertex representing the subset to each one of the $w$ vertices representing the element if the element is contained in the subset in the max-k-coverage instance, and set the weight of this edge to be $1$. Seeding a vertex representing a subset simulates selecting a subset in the max-k-coverage problem. This is obviously an approximation-ratio preserving reduction by setting $w$ to be sufficiently large. Given that the inapproximability ratio $(1-1/e+\varepsilon)$ of max-k-coverage is known due to Feige~\citep{Feige98}, we have the same inapproximability ratio for \InfMax.\footnote{This is exactly the reduction used in Kempe et al.~\cite{KKT15}.} In addition, the above-mentioned graph is a directed acyclic graph, so the same inapproximability ratio extends to the \InfMax problem with directed acyclic graphs.
The inapproximability does not come from the $\#$P-hardness of evaluating $\sigma(\cdot)$, as this function is easy to compute in the above instance since all edges have weight $1$.

We now complement this result by showing that the \InfMax problem on directed acyclic graphs is also APX-hard under the \LT model.

\begin{theorem}[APX-Hardness on DAGs]
\label{thm:apx-hardness}
There exists a universal constant $\tau>0$ such that \InfMax on directed acyclic graphs with the \LT model is NP-hard to approximate to within a factor of $(1-\tau)$.
\end{theorem}

We will use the following inapproximability result for the vertex cover problem.
Theorem~\ref{thm:VC} is a straightforward corollary to Dinur's PCP theorem~\cite{dinur2007pcp}, which states that max-3SAT is APX-hard for instances where each variable appears a constant number of times.
Based on Dinur's PCP theorem, Theorem~\ref{thm:VC} follows from the well-known reduction from 3SAT to vertex cover.

\begin{theorem}\textup{\textbf{(APX-Hardness of Vertex Cover)}}\label{thm:VC}
    There exist a universal constant integer $d$ and a universal constant $\gamma\in(0,1)$ such that, given an even integer $k$ and an undirected graph $G=(V,E)$ where the degree of each vertex is bounded by $d$ and $|V|=\frac32k$, it is NP-hard to distinguish between the following two cases:
    \begin{itemize}
        \item YES: $G$ has a vertex cover of size $k$;
        \item NO: all vertex covers of $G$ have sizes at least $(1+\gamma)k$.
    \end{itemize}
\end{theorem}
\begin{proof}
We prove this theorem by reducing the hard instance of 3SAT, derived from Dinur's PCP theorem, to the Vertex Cover problem on bounded-degree graphs.

First, by Dinur's PCP theorem, there exists a constant $\epsilon_0 > 0$ such that it is NP-hard to distinguish whether a 3CNF formula $\phi$ is satisfiable (YES) or if at most $(1-\epsilon_0)$ fraction of clauses in $\phi$ can be satisfied (NO), even if each variable appears in at most $d'$ clauses (where $d'$ is a universal constant).

Then we construct a reduction from this 3SAT instance $\phi$ to a graph $G=(V,E)$. Let $\phi$ have $n'$ variables and $m'$ clauses. The construction is as follows. For each clause $C_j = (l_1 \lor l_2 \lor l_3)$, create a triangle of three nodes $c_{j,1}, c_{j,2}, c_{j,3}$.
We connect $c_{j_1,t_1}$ and $c_{j_2,t_2}$ (where $j_1,j_2\in\{1,\ldots,m'\}$ and $t_1,t_2\in\{1,2,3\}$) if the corresponding two literals are exactly $x_i$ and $\neg x_i$ for the same variable $x_i$.
Finally, set $k=2m'$.
The graph $G=(V,E)$ constructed satisfies that 1) $|V|=3m'=\frac32k$, and 2) the degree of each vertex is at most $2+d'$, which satisfies the descriptions in the theorem statement.

If $\phi$ is a yes instance, we have a satisfying assignment of $\phi$.
Under this assignment, each clause has at least one literal that evaluates to true.
We choose one such representative literal from each clause, and we obtain an independent set of size $m'$.
Its complement forms a vertex cover of size $2m'=k$.

If $\phi$ is a no instance, any assignment leaves at least $\epsilon_0m'$ clauses unsatisfied. 
As a result, the maximum independent set has size at most $(1-\epsilon_0)m'$.
Thus, the minimum vertex cover has size at least $2m'+\epsilon_0m'=(1+\frac{\epsilon_0}2)k$.
Setting $\gamma=\frac{\epsilon_0}2$ concludes the theorem.
\end{proof}

As a few remarks, the theorem above requires that the degrees of the vertices are bounded by a constant $d$, which also implies $|E|=O(|V|)$ ($|E|\leq\frac d2|E|$ to be more precise).
In addition, the optimal size of the vertex cover in the yes instance is a constant fraction of the total number of vertices: $k$ versus $|V|=\frac32k$.


We prove Theorem~\ref{thm:apx-hardness} in the remaining part of this section.

\subsection{Construction and Preliminary Observations}
Given a Vertex Cover instance $(G=(V,E),k)$ with $|V|=n=\frac32k$ and $|E|=m\leq\frac d2n$ that satisfies the description in Theorem~\ref{thm:VC}, we construct an influence maximization instance $G'$ where the underlying graph is a DAG. The core idea of the construction is to design a local gadget for every edge in the original graph, such that the selection of endpoints in the original graph directly corresponds to the influence gain obtainable within the gadget.

\begin{figure}[htbp]
    \centering
    \includegraphics[width=0.7\linewidth]{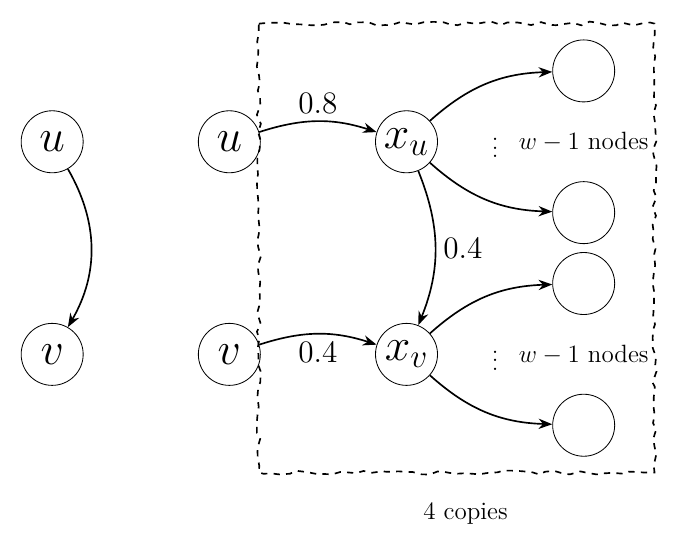}
    \caption{Reduction gadget: Converting an edge $(u,v)$ of the original graph into a gadget for the influence maximization instance.}
    \label{fig:reduction}
    \vspace{-0.5cm}
\end{figure}

Specifically, for each undirected edge $(u,v) \in E$ in the original graph, we replace the edge by the gadget shown in \Cref{fig:reduction} into $G'$. The gadget contains two intermediate nodes $x_u$ and $x_v$, corresponding to the two endpoints of the edge $(u,v)$. A directed edge with weight $b_1 = 0.8$ is added from $u$ to $x_u$, a directed edge with weight $b_3 = 0.4$ is added from $v$ to $x_v$, and a directed edge with weight $b_2 = 0.4$ is added from $x_u$ to $x_v$. Furthermore, $w-1$ leaf nodes are attached to both $x_u$ and $x_v$; these leaf nodes are activated with probability $1$ whenever their respective parent node is activated, thereby amplifying the influence differences within the gadget.
Finally, this gadget is duplicated three more times (the vertices $u$ and $v$ in the original graph are not duplicated, the remaining part of the gadget has four identical copies).
By repeating this process for all edges, we obtain the influence maximization instance $G'$, which is globally a directed acyclic graph.

Finally, we set the number of seeds $k'$ to be $k'=k+4m$.
We obtain an \InfMax instance $(G',k')$.

Below, we list some preliminary observations.
\begin{enumerate}
    \item If we choose both intermediate-layer vertices $x_u$ and $x_v$ as seeds, all vertices in each copy of the gadget are infected with probability $1$. The expected number is $2w$.
    \item If one or both of $u$ and $v$ is/are seeded and we carefully choose exactly one of $x_u,x_v$ as a seed, then the expected number of infected vertices in each copy of the gadget reaches $1.8w$. To see this, if $u$ is seeded, we choose $x_v$ as a seed, then $x_u$ is infected with probability $0.8$ and $x_v$ is infected with probability $1$; if $v$ is seeded, we choose $x_u$ as a seed, then $x_u$ is infected with probability $1$ and $x_v$ is infected with probability $0.4+0.4=0.8$; if both $u$ and $v$ are seeded, the choice of $x_u$ and $x_v$ can be arbitrary. In all cases, the expected number of infected vertices is $1.8w$.
    \item If one or both of $u$ and $v$ is/are seeded and none of $x_u$ and $x_v$ is seeded, the expected number of infected vertices in each copy of the gadget is at most $1.52w$. To see this, even if both $u$ and $v$ are seeded, without seeding $x_u$ and $x_v$, $x_u$ is infected with probability $0.8$, and $x_v$ is infected with probability $0.4+0.8\times0.4=0.72$. The expected number of infected vertices is at most $(0.8+0.72)w=1.52w$.
    \item If none of $u$ and $v$ is seeded (meaning that the edge $(u,v)$ in the vertex cover instance is uncovered) and we are to choose one seed from $x_u$ and $x_v$, the best choice is $x_u$, and the expected number of infected vertices is $(1+0.4)w=1.4w$.
\end{enumerate}

\paragraph{Completeness.}
Suppose there exists a vertex cover $S$ of size $k$ in the original graph $G$. We construct a seed set $S'$ in $G'$ as follows: First, include all nodes from $S$ in $S'$. Second, for each edge $(u,v) \in E$, since $S$ is a vertex cover, at least one endpoint belongs to $S$. If $u \in S$, we add node $x_v$ to $S'$; if $v \in S$, we add node $x_u$ to $S'$. If both $u,v \in S$, we choose one arbitrarily. We do this for all four copies in the gadget.
The number of seeds chosen is exactly $k'=k+4m$.
By Observation 2, each copy of the gadget has $1.8w$ infected vertices in expectation.
Since each gadget has four copies and there are $m$ edges, the total expected number of infected vertices is
$k+4m\cdot 1.8w>7.2mw.$

\paragraph{Soundness.}
Now, suppose the size of any vertex cover in $G$ is at least $(1+\gamma)k$.
We consider the optimal choice $S'$ of the seed set in the \InfMax instance $(G',k')$, and we aim to show that $\sigma(S')$ is at most $(1-\delta)7.2mw$ for some constant $\delta$. 
We first give some characterizations for $S'$.

First of all, notice that $S'$ cannot contain a seed from those $w-1$ leaf nodes (those vertices on the right-hand side of Fig.~\ref{fig:reduction}): seeding the corresponding $x_u$ or $x_v$ is much better.
Next, we show that there cannot exist a copy of the gadget where both intermediate vertices $x_u$ and $x_v$ are seeded in $S'$.

\paragraph{Sub-optimality of Doubly-Seeding.}
Suppose there exists a vertex pair $(u,v)$ such that both $x_u$ and $x_v$ are seeded in one copy of the gadget.
We will show that such a seeding strategy is always sub-optimal.
We consider two cases regarding the set of seeds chosen among the vertices in the original vertex cover instance, i.e., $S'\cap V$:
\begin{itemize}
    \item Case 1: $S'\cap V$ is a vertex cover of $G$, and
    \item Case 2: $S'\cap V$ is not a vertex cover of $G$.
\end{itemize}

For Case 1, we must have $|S'\cap V|\geq(1+\gamma)k$.
Since $k'=k+4m$, we can choose at most $4m-\gamma k<4m$ vertices among those intermediate vertices in the gadgets.
If there is one copy $(u,v)$ where both $x_u$ and $x_v$ are seeded, then there must exist another copy $(u',v')$ where none of $x_{u'}$ and $x_{v'}$ is seeded.
By replacing a seed from the former copy with a seed in the latter copy, the expected number of infected vertices in the graph increases.
To see this, by Observations 1 and 2, the extra seed in the former copy has a $0.2w$ marginal influence; by Observations 2 and 3, adding one seed to the latter copy increases the expected influence by at least $1.8w-1.52w=0.28w$.
Hence, the exchange gives an increment of $0.08w$ in expected influence.
Thus, doubly-seeding is always sub-optimal under Case~1.

We then consider Case 2.
Consider an arbitrary uncovered edge $(u,v)$ in the original graph.
We first show that all the $8$ intermediate vertices in the gadget must be seeded.
Suppose this is not the case.
If exactly $7$ intermediate vertices in the gadget are seeded, then replacing one of them by either $u$ or $v$ improves the expected influence.
Before the replacement, by Observations 1 and 4, the expected number of infected vertices in the gadget is $2w+2w+2w+1.4w=7.4w$.
After the replacement, by Observations 1 and 2, the expected number of infected vertices in the gadget is $2w+2w+1.8w+1.8w=7.6w$.
Thus, this replacement is beneficial.
If at most $6$ intermediate vertices in the gadget are seeded, then replacing any double seed (from this gadget or from another gadget) by either $u$ or $v$ improves the expected influence.
By adding one of $u$ and $v$ as seeds and by Observation 2 and 4, the improvement of influence within this gadget is at least $2(1.8w-1.4w)=0.8w$.
On the other hand, the loss in the expected influence due to selecting one less double seed is at most $2w-1.4w=0.6w$ (by Observation 1 and 4).
This replacement is also beneficial.

Next, if there are at least $3$ copies (from the same gadget, or from different gadgets) where none of the two intermediate vertices is seeded, we show that the existence of an uncovered edge $(u,v)$ where all the $8$ intermediate vertices are seeded makes $S'$ sub-optimal.
We can instead select only $4$ intermediate vertices such that each of the four copies contains only one seeded intermediate vertex, and, for the $4$ more seeds, we use one of them to seed $u$ or $v$, and the other $3$ to seed the intermediate vertices from the $3$ copies where no intermediate seed is chosen.
To show that this adjustment is beneficial, we first compute the loss in the expected influence within the gadget $(u,v)$.
Before the adjustment, by Observation 1, the expected influence is $2w\times 4=8w$.
After the adjustment, by Observation 2, it is $1.8w\times 4=7.2w$, the loss is $0.8w$.
On the other hand, for the $3$ seeds put in the other copies, each of them introduces a marginal influence of at least $1.8w-1.52w=0.28w$.
Therefore, the overall gain is $0.28w\times 3=0.84w>0.8w$.
This adjustment is beneficial.

Finally, it remains the case where there are at most $2$ copies where none of the two intermediate vertices is seeded, and, in addition, for every edge $(u,v)$ in the original graph that is not covered, all the $8$ intermediate vertices are seeded.
It is easy to see that, in this case, we have used more than $k'=k+4m$ seeds, leading to a contradiction.
This can be proved by a single counting argument.
Consider the following rearrangement of the seeds: choose any uncovered edge $(u,v)$; among the $8$ intermediate seeds, choose (at most) $2$ of them to compensate for the (at most) $2$ copies where no intermediate vertices are seeded; after this, for each uncovered edge $(u,v)$, at least $6$ intermediate vertices are still seeded; then we replace one of the intermediate seed by either $u$ or $v$ to make $(u,v)$ covered.
By this rearrangement, all the edges are now covered, and each copy still contains at least one intermediate seed.
The total number of seeds is at least $(1+\gamma)k+4m>k'$.
This concludes the analysis of Case 2.

\paragraph{Bounding $\sigma(S')$.}
Now we compute an upper bound to the expected influence of $S'$.
First of all, since we have ruled out double-seeding, in each of the four copies in each gadget, the expected number of infected vertices is at most $1.8w$.

We consider two cases regarding $S'\cap V$:
\begin{itemize}
    \item Case 1: $|S'\cap V|>(1+\frac\gamma2)k$, and
    \item Case 2: $|S'\cap V|\leq (1+\frac\gamma2)k$.
\end{itemize}

For the first case, given $k'=k+4m$, there are at least $\frac\gamma2k$ copies where no intermediate vertex is seeded.
Even if each of the remaining $4m-\frac\gamma2k$ copies has the maximum number $1.8w$ of infected vertices, the expected total number of infected vertices in the gadgets is at most
\begin{equation}\label{eqn:soundness1}
    1.8w\cdot\left(4m-\frac\gamma2k\right)+\frac\gamma2k\cdot1.52w=7.2mw-0.14\gamma kw.
\end{equation}

For the second case, since a vertex cover has size at least $(1+\gamma)k$, at least $\frac\gamma2k$ edges are not covered.
Without double-seeding, the expected number of infected vertices in each gadget is at most $1.4w\times 4=5.6w$ (Observation 4).
The expected total number of infected vertices in the gadgets is at most
\begin{equation}\label{eqn:soundness2}
    1.8w\cdot4\cdot\left(m-\frac\gamma2k\right)+5.6w\cdot\frac\gamma2k=7.2mw-0.8\gamma kw.
\end{equation}

Combining (\ref{eqn:soundness1}) and (\ref{eqn:soundness2}), the upper bound to the expected number of infected vertices in the gadgets is
$
    7.2mw-0.14\gamma kw.
$
By making $w$ sufficiently large, e.g., $w=n^2$, the number of infected vertices outside those gadgets, i.e., those $n$ vertices in the original graph, is negligible.
We can conclude that the expected number of infected vertices in the entire graph is at most $7.2mw-0.14\gamma kw+o(w)$.

Finally, we conclude Theorem~\ref{thm:apx-hardness} by showing that the gap between $7.2mw$ and $7.2mw-0.14\gamma kw+o(w)$ is constant.
Noticing that $\gamma$ is a constant, we only need to show that $k>cm$ for some constant $c$.
This is true, as we have $k=\frac23n$ and $m\leq\frac d2n$.

\section{Influence Maximization on Out-Arborescence}\label{subsec:out-arborescence}

Let $T=(V, E)$ be an out-arborescence where edges are directed from parents to children. In the case of out-arborescences, we demonstrate that the unique path property facilitates a polynomial-time solution via dynamic programming. 
 \begin{theorem}\label{thm:out}
     For the influence maximization problem on out-arborescences, given any seed set $S$, we always have $\sigma_\IC(S)=\sigma_\LT(S)$. In addition, there exists a polynomial-time algorithm to compute $S$ that maximizes $\sigma_\IC(S)=\sigma_\LT(S)$.
 \end{theorem}
 

We first present the following lemma, which establishes the equivalence between the two models.
The proof is straightforward from the definitions of the two models, and is thus omitted.

\begin{lemma}\label{lemma:propagation-probability}
In an out-arborescence $T$ under both the \IC and \LT models, for any nodes $u$ and $v$, if there exists a unique directed path $p = \langle u = w_0, w_1, \ldots, w_k = v \rangle$, then
$
h_{u,v} = \prod_{i=0}^{k-1} P_{(w_i, w_{i+1})}
$
where $P_{(w_i, w_{i+1})}$ represents the weight (propagation probability) of edge $(w_i, w_{i+1})$ and $h_{u,v}$ denote the probability that $u$ activates $v$ (if $u$ is the unique seed in $T$). If no directed path exists, $h_{u,v} = 0$.
\end{lemma}

\Cref{lemma:propagation-probability} gives an effective way to compute the probability of infection for each node $u$ given a seed set $S$.
If $u$ does not have any seeded ancestor, then $u$ is infected with probability $0$.
Otherwise, $u$'s infection probability only depends on the position of the lowest seeded ancestor $s$ (any other seeds, including seeds on the other branches, or ancestor of $s$, play no role in infecting $u$).
Then, $u$'s infection probability is $h_{s,u}$ according to \Cref{lemma:propagation-probability}.
We can compute this for all vertices, and the expected total number of infected vertices is just a summation over all vertices' infection probabilities by linearity of expectation.
This gives an effective way to compute $\sigma(S)$ for a seed set $S$.
Moreover, this holds for both the \IC and \LT models.
Therefore, the first part of Theorem~\ref{thm:out} concludes.

We now consider \InfMax on an out-arborescence: Given $T=(V,E)$, edge weights $\{b_{(u,v)}\}$, and a budget $K$, select $K$ seeds to maximize the expected influence.
We first convert $T$ to a \emph{binary tree} by introducing additional \emph{virtual} nodes.
For example, if $u$ has three children $v_1,v_2,v_3$, we introduce an extra node $u'$ such that $u$ has two children $u'$ and $v_3$, and then $u'$ is the parent of $v_1$ and $v_2$.
Set $p_{u',v_1}$ and $p_{u',v_2}$ to be $p_{u,v_1}$ and $p_{u,v_2}$ in the original graph.
Set $p_{u,u'}=1$.
We can do this recursively for a node with more than three children.
Notice that those virtual nodes cannot be seeded. 
In the remaining part of this section, we assume $T$ is a binary tree with virtual nodes.

\begin{definition}[Dynamic Programming State]\label{def:dp-state}
For any node $v \in V$ in tree $T=(V,E)$, define the DP state as $I(v,k,p)$,
representing the maximum expected influence obtainable within the subtree rooted at $v$, given that at most $k$ seeds are selected within this subtree and node $v$ receives an external influence intensity $p \in [0,1]$ from its ancestors.
\end{definition}

The parameter $p$ captures the influence transmitted from the parent or the ancestors to $v$. If $v$'s parent $u$ is selected as a seed, $p = p_{(u,v)}$. If $u$ is not a seed but is activated by its own ancestors, the influence decays along the edge, i.e., $p = p_{(u,v)} \cdot p_u$. For the root, $p=0$.
We distinguish two mutually exclusive cases for any internal node $v$:

\textbf{$v$ is not selected as a seed ($I_{\mathrm{no}}$).} The activation of $v$ depends entirely on the external influence $p$, so $v$ contributes $p$ to the expected influence if $v$ is not virtual. All $k$ seeds must be distributed among its children's subtrees. For a binary tree with children $v_l$ and $v_r$:
\[
I_{\mathrm{no}}(v,k,p) = \max_{0\le i\le k} \Bigl\{ I\bigl(v_l,i,p\cdot b_{(v,v_l)}\bigr) + I\bigl(v_r,k-i,p\cdot b_{(v,v_r)}\bigr) + p\cdot\mathbb{I}(v\mbox{ is not virtual}) \Bigr\},
\]
where $\mathbb{I}$ is the indicator function.

\textbf{$v$ is selected as a seed ($I_{\mathrm{sel}}$).} Node $v$ is then not virtual and is activated with certainty, contributing 1. The remaining $k-1$ seeds are distributed in the subtrees. The children receive influence directly from seed $v$:
\[
I_{\mathrm{sel}}(v,k,p) = \max_{0\le i\le k-1} \Bigl\{ I\bigl(v_l,i,b_{(v,v_l)}\bigr) + I\bigl(v_r,k-1-i,b_{(v,v_r)}\bigr) + 1 \Bigr\}.
\]

The complete transition equation is:
\begin{equation}\label{eq:dp-transition}
I(v,k,p) = 
\left\{\begin{array}{ll}
    \max\bigl\{ I_{\mathrm{no}}(v,k,p),\; I_{\mathrm{sel}}(v,k,p) \bigr\} &  \mbox{if }v\mbox{ is not virtual}\\
   I_{\mathrm{no}}(v,k,p)  & \mbox{otherwise}
\end{array}
\right..
\end{equation}

Then we define the boundary conditions for a leaf node $v$: $I(v,k,p)=1$ if $k\geq1$, and $I(v,k,p)=p$ otherwise.

Finally, to ensure the validity of the dynamic programming algorithm, we need to make sure that we only need to compute a polynomial number of $I(v,k,p)$.
For the third parameter $p$, we need to make sure, for each node $v$, only a polynomial number of values for $p$ is needed.
To see this, we need to compute $I(\mbox{root},K,0)$ at the end of the day.
For a vertex $u$ at the second level, we only need $I(u,k,p)$ for $k=0,1,\ldots,K$ and $p\in\{0,p_{\mbox{root},u}\}$.
For a vertex $u$ at the third level with parent $v$, we only need $I(u,k,p)$ for $k=0,1,\ldots,K$ and $p\in\{0,p_{v,u},p_{v,u}\cdot p_{\mbox{root},u}\}$.
Similarly, for a vertex $u$ at the $i$-th level, we only need $(i+1)$ values for $p$.
Therefore, the number of values of $p$ needed for each $I(v,k,p)$ is bounded by the depth of $v$, which is bounded by $O(n)$.
The total number of values for $I(v,k,p)$ is bounded by $O(n^2K)$.
Computing each recurrence relation requires $O(K)$ time (for enumerating $i$).
Thus, the overall time complexity is $O(n^2K^2)$.


\subsection{Details and Correctness of the Algorithm for Out-Arborescence}\label{sec:corr}

Our dynamic programming algorithm is presented in Algorithm~\ref{alg:DPIM_recursive}.

\begin{algorithm}[!h]
\caption{Recursive Dynamic Programming with Memoization}
\label{alg:DPIM_recursive}

\SetKwProg{Fn}{Function}{ is}{end}
\SetKwFunction{FCompute}{Compute-I}

\BlankLine
\KwIn{Out-arborescence $T=(V,E)$, budget $K$, edge weights $B$}
\KwOut{Maximum expected influence: \FCompute{$root, K, 0$}}
\BlankLine

Initialize a global hash map $\mathit{Memo} \gets \emptyset$\;
\BlankLine

\Fn{\FCompute{$v, k, p$}}{
    \tcp{Check if the state $(v, k, p)$ has been computed}
    \If{$(v, k, p) \in \mathit{Memo}$}{
        \Return $\mathit{Memo}[(v, k, p)]$\;
    }
    
    \eIf{$v$ is a leaf}{
        \tcp{Base case: influence at a leaf node}
        $res = (k \geq 1) ? 1 : p$\;
    }{
        \tcp{Case 1: $v$ is not selected as a seed}
        $I_{\text{no-select}} = \max_{0 \le i \le k} \{ \FCompute{$v_l, i, p \cdot b_{(v,v_l)}$} + \FCompute{$v_r, k-i, p \cdot b_{(v,v_r)}$} + p\cdot\mathbb{I}(v\mbox{ is not virtual}) \}$\;
        
        \tcp{Case 2: $v$ is selected as a seed}
        \If{$k \ge 1$}{
            $I_{\text{select}} = \max_{0 \le i \le k-1} \{ \FCompute{$v_l, i, b_{(v,v_l)}$} + \FCompute{$v_r, k-1-i, b_{(v,v_r)}$} + 1 \}$\;
        }
        $res = \max \{ I_{\text{no-select}}, I_{\text{select}} \}$ if $v$ is not virtual and $k\geq1$, and $res=I_{\text{no-select}}$ otherwise\;
    }
    
    \tcp{Store the computed result in Memo}
    $\mathit{Memo}[(v, k, p)] \gets res$\;
    \Return $res$\;
}

\BlankLine
\Return{\FCompute{$root, K, 0$}}\;
\end{algorithm}

The following lemma proves the correctness of our algorithm.

\begin{lemma}[Correctness of \Cref{alg:DPIM_recursive}]\label{theorem:dpim-correctness}
For any out-arborescence $T$ and budget $K$, \Cref{alg:DPIM_recursive} returns the maximum expected influence.
\end{lemma}

\begin{proof}
We prove the correctness of the recursive function $\textsc{Compute-I}(v, k, p)$ by induction on the height $H$ of the sub-arborescence rooted at $v$. Let $I^*(v, k, p)$ denote the true maximum expected influence in the sub-arborescence rooted at $v$ under the Independent Cascade (IC) model, given budget $k$ and parent activation probability $p$.

\textbf{Base Case:} Suppose $H = 0$, meaning $v$ is a leaf node. 
\begin{itemize}
    \item If $k \ge 1$, the optimal strategy is to select $v$ as a seed. Since a seed is activated with probability 1, the expected influence is 1. The algorithm returns $(k \ge 1) ? 1 : p$, which is 1.
    \item If $k = 0$, $v$ cannot be a seed. Its activation depends solely on the probability $p$ from its parent. The expected influence is $p$. The algorithm returns $p$.
\end{itemize}
Thus, $\textsc{Compute-I}(v, k, p) = I^*(v, k, p)$ for the base case.

\textbf{Inductive Hypothesis:} Assume that for all sub-arborescences with height less than $H$, the function $\textsc{Compute-I}$ correctly returns the maximum expected influence.

\textbf{Inductive Step:} Consider a sub-arborescence rooted at $v$ with height $H > 0$. Let $v_l$ and $v_r$ be the left and right children of $v$ with edge weights $b_{(v, v_l)}$ and $b_{(v, v_r)}$, respectively. The optimal strategy must fall into one of two cases:

\textit{Case 1: $v$ is not selected as a seed.} 
The node $v$ is activated with probability $p$. If activated, it independently attempts to activate $v_l$ and $v_r$ with probabilities $b_{(v, v_l)}$ and $b_{(v, v_r)}$. Thus, the conditional activation probabilities passed to children are $p \cdot b_{(v, v_l)}$ and $p \cdot b_{(v, v_r)}$. By the inductive hypothesis, the maximum influence from sub-arborescences rooted at $v_l$ and $v_r$ with budgets $i$ and $k-i$ are correctly computed as $\textsc{Compute-I}(v_l, i, p \cdot b_{(v, v_l)})$ and $\textsc{Compute-I}(v_r, k-i, p \cdot b_{(v, v_r)})$. The total expected influence is the sum of $v$'s activation probability $p$ and the optimal distribution of budget $k$ among its children:
\[ I_{\text{no-select}} = \max_{0 \le i \le k} \{ \textsc{Compute-I}(v_l, i, p \cdot b_{(v,v_l)}) + \textsc{Compute-I}(v_r, k-i, p \cdot b_{(v,v_r)}) + p \} \]

\textit{Case 2: $v$ is selected as a seed ($k \ge 1$).} 
The node $v$ is activated with probability 1. Consequently, the activation probabilities passed to children are simply $b_{(v, v_l)}$ and $b_{(v, v_r)}$, independent of the original $p$. With $k-1$ budget remaining, the total expected influence is:
\[ I_{\text{select}} = \max_{0 \le i \le k-1} \{ \textsc{Compute-I}(v_l, i, b_{(v,v_l)}) + \textsc{Compute-I}(v_r, k-1-i, b_{(v,v_r)}) + 1 \} \]

The algorithm computes $res = \max \{ I_{\text{no-select}}, I_{\text{select}} \}$. Since it exhaustively searches the binary choice of seeding $v$ and all possible partitions of the budget $k$ across the unique paths in the out-arborescence, and since all sub-problems are optimal by the inductive hypothesis, $res$ must equal $I^*(v, k, p)$. By the principle of mathematical induction, the algorithm is correct for any out-arborescence of finite height.
\end{proof}

\section{Influence Maximization on In-Arborescence Structures}
\label{subsec:in-arborescence}

In this section, we address the \InfMax problem under the Independent Cascade (IC) model on in-arborescence structures.  While it is known that the problem is solvable in polynomial time under the Linear Threshold (LT) model, it remains NP-hard under the IC model \cite{wang2016bharathi,lu2017solution}. We bridge this gap by establishing the following theorem:

\begin{theorem}
    There exists a Fully Polynomial Time Approximation Scheme (FPTAS) for the influence maximization problem under the IC model on an in-arborescence.
\end{theorem}

Let $T=(V,E)$ be an in-arborescence where edges are directed from children to parents.
As introduced in \cref{subsec:out-arborescence}, we first transform the in-arborescence $T$ into a binary tree by introducing auxiliary \emph{virtual nodes}. 
Consider a node $u$ with three children $v_1, v_2,$ and $v_3$. We introduce a virtual node $u'$ such that $u$ becomes the parent of $u'$ and $v_3$, while $u'$ becomes the parent of $v_1$ and $v_2$. To preserve the original propagation probabilities, we set the edge weights $(v_1, u')$ and $(v_2, u')$ equal to the original $b_{v_1,u}$ and $b_{v_2,u}$ respectively, and assign a weight of $1$ to the edge $(u', u)$. This process is applied recursively to any node with a degree greater than two. It is important to note that these virtual nodes are purely structural and cannot be selected as seeds. In the subsequent analysis, we assume $T$ refers to this transformed binary structure containing virtual nodes.

Let $v_l, v_r$ be the left and right children of $v$, with edge weights $b_l, b_r$ and activation probabilities $p_l, p_r$, correspondingly. The probability that $v$ is activated by its children is given by the non-linear combination:
\begin{equation}
    p_v = 1 - (1 - p_l \cdot b_l)(1 - p_r \cdot b_r)
\end{equation}

Within a subtree $T_v$ rooted at $v$, two quantities are of interest: the total expected influence $\varsigma(v)$ within the subtree $T_v$ and the activation probability $p_v$ of vertex $v$.
Locally, we would like to maximize $\varsigma(v)$. 
Globally, $p_v$ completely controls the influence from $T_v$ to the remaining part of the graph.

To handle the trade-off between local influence $\varsigma$ and upward activation probability $p$, we propose a dynamic programming approach based on the Pareto Frontier. Instead of a single value, we maintain a set of pairs:
\begin{equation}
    \mathcal{S}(v, k) = \{ (p, \varsigma)\}
\end{equation}
where $p$ denotes the probability of infecting $v$ and $\varsigma$ represents the expected number of infected nodes in $T_v$ is $\varsigma$. Each pair $(p, \varsigma)$ represents a strategy using at most $k$ seeds in $T_v$. 

To guarantee the size of $\mathcal{S}(v, k)$ is polynomially bounded, we need to consider roundings for $p$ and $\varsigma$.
Choose a sufficiently large integer $M$ and let $\delta=1/M$. We will round the values of $p$ and $\varsigma$ downward such that they are integer multiples of $\delta$.
Therefore, the range of $p$ is $\{0,\delta,2\delta,\ldots,1-\delta,1\}$, with size $M+1$, and the range of $\varsigma$ is a set of size $nM+1$ (since the number of infected vertices cannot exceed $n$).
Thus, $|\mathcal{S}(v, k)|\leq (M+1)(nM+1)$.
An alternative way to view $\mathcal{S}(v, k)$ is a $(M+1)\times(nM+1)$ Boolean table, where the index of each cell is a pair $(p,\varsigma)$ and the Boolean value indicates if there is a seeding strategy such that $v$ is infected with probability approximately $p$ and the expected number of infected vertices in the subtree $T_v$ is approximately $\varsigma$.

When calculating $\mathcal{S}(v, k)$ for an internal node $v$, we define two operators, $\mathcal{I}_{sel}$ and $\mathcal{I}_{no}$, representing whether $v$ is selected as a seed. We iterate through all possible budget partitions $i \in \{0, \dots, k\}$ to capture all achievable states.

\begin{enumerate}
    \item \textbf{Case 1: $v$ is selected as a seed ($\mathcal{I}_{sel}$):} This case only applies if $v$ is a \textbf{real node} and $k \ge 1$. Here, $v$ is activated with certainty ($p = 1.0$), contributing $1.0$ to the influence. For each budget split $i \in \{0, \dots, k-1\}$, we combine all pairs from children:
    \[ \varsigma_{tmp} = \varsigma_l + \varsigma_r + 1.0 \]
    The discretized pair $(1.0, \text{Round}(\varsigma_{tmp}))$ is added to $\mathcal{S}(v, k)$. We do not take the maximum here; we retain all rounded pairs to allow for optimal combinations at higher levels of the tree.

    \item \textbf{Case 2: $v$ is not selected as a seed ($\mathcal{I}_{no}$):} This applies to both real and virtual nodes. For each $i \in \{0, \dots, k\}$, we consider combinations $(p_l, \varsigma_l) \in \mathcal{S}(v_l, i)$ and $(p_r, \varsigma_r) \in \mathcal{S}(v_r, k-i)$. The activation probability is:
    \[ p_{tmp} = 1 - (1 - p_l \cdot b_l)(1 - p_r \cdot b_r) \]
    The expected influence $\varsigma_{tmp}$ is calculated as:
    \begin{itemize}
        \item If $v$ is a \textbf{real node}: $\varsigma_{tmp} = \varsigma_l + \varsigma_r + p_{tmp}$.
        \item If $v$ is a \textbf{virtual node}: $\varsigma_{tmp} = \varsigma_l + \varsigma_r$.
    \end{itemize}
    Every resulting pair $(\text{Round}(p_{tmp}), \text{Round}(\varsigma_{tmp}))$ is added to $\mathcal{S}(v, k)$.
\end{enumerate}
The complete transition for a node $v$ with budget $k$ is established by aggregating all achievable pairs from its children and applying immediate rounding to maintain a bounded state space. Formally, the set of achievable pairs $\mathcal{S}(v, k)$ is constructed as:
\begin{equation}\mathcal{S}(v, k) = \text{Round} \left( \bigcup_{0 \le i \le k} \mathcal{I}_{no}(v, i, k-i) \cup \bigcup_{0 \le i \le k-1} \mathcal{I}_{sel}(v, i, k-1-i) \right)
\end{equation}where $\text{Round}(\cdot)$ denotes the element-wise downward rounding of each pair $(p, \varsigma)$ to the discretized grid specified by $\delta$. The details of this procedure are listed in \Cref{alg:DPIM_InArb_Final}.To establish the $(1-\epsilon)$-approximation guarantee, we first prove that the recursive structure of our dynamic programming approach correctly captures all feasible $(p, \varsigma)$ pairs in the absence of discretization. We then analyze the cumulative error introduced by the rounding operations at each node. Specifically, we show that by setting the discretization step $\delta = \mathcal{O}(\epsilon/n^2)$, the total additive error across the in-arborescence is bounded by $\epsilon$, while the size of each set $\mathcal{S}(v, k)$ remains $O(nM^2)$.

\begin{lemma}[Optimality of the Recursive Structure]
    In the absence of the rounding step, \Cref{alg:DPIM_InArb_Final} returns the exact set of all achievable pairs $\mathcal{S}(v, k)$ for any subtree $T_v$ and budget $k$.
\end{lemma}

\begin{proof}
    We prove this by induction on the height of the in-arborescence $T$.
    
    \textbf{Base Case:} Let $v$ be a leaf node. The algorithm returns $\mathcal{S}(v, 0)=\{(0.0, 0.0)\}$ and $\mathcal{S}(v, 1)=\{(1.0, 1.0)\}$. These represent the only feasible states for a leaf, which is exactly the set of all achievable pairs.
    
    \textbf{Inductive Step:} Assume that for any child $c \in \{v_l, v_r\}$, $\mathcal{S}(c, i)$ contains all exact achievable pairs.
    \begin{itemize}
        \item \textbf{Case $I_{no}$ ($v$ is not a seed):} The algorithm explores all budget partitions $i$ and all pairs $(p_l, \varsigma_l) \in \mathcal{S}(v_l, i)$ and $(p_r, \varsigma_r) \in \mathcal{S}(v_r, k-i)$. In an in-arborescence, the activation of $v$ by its children follows the independent trial formula $p_{tmp} = 1 - (1 - p_l b_l)(1 - p_r b_r)$. Since the Cartesian product covers all possible combinations from the children, the resulting set contains all possible $(p, \varsigma)$ pairs for node $v$.
        \item \textbf{Case $I_{sel}$ ($v$ is a seed):} By selecting $v$ as a seed, $p_v=1.0$. The algorithm iterates through all combinations of achievable pairs from subtrees to compute $\varsigma_{tmp} = \varsigma_l + \varsigma_r + 1$. By induction, this covers all possible expected influence values when $v$ is seeded.
    \end{itemize}
    Thus, without rounding, $\mathcal{S}(v, k)$ contains all exact achievable pairs.
\end{proof}

\begin{lemma}[Error Bound]
If the state space is discretized with a step size $\delta = 1/M$, the total additive error in the expected influence at the root is at most $n^2 \delta$.
\end{lemma}

\begin{proof}
By the linearity of expectation, the expected influence $\varsigma$ is the sum of the activation probabilities of all nodes: $\varsigma = \sum_{u \in V} p_u$. Let $p_u$ be the exact probability and $\hat{p}_u$ be the value computed under discretization. At each node $v$, the \text{Round} operation introduces a local rounding error of at most $\delta$.

The probability combination function $f(p_l, p_r) = 1 - (1 - p_l b_l)(1 - p_r b_r)$ has partial derivatives $\frac{\partial f}{\partial p_l}, \frac{\partial f}{\partial p_r} \le 1$. Let $e_v = |p_v - \hat{p}_v|$ be the error at node $v$. This error is bounded by the sum of its children's errors plus the local rounding error $\delta$:
$$ e_v \le \frac{\partial f}{\partial p_l}e_l + \frac{\partial f}{\partial p_r}e_r + \delta \le e_l + e_r + \delta $$
By induction, the error $e_u$ at any node $u$ is bounded by the number of nodes in its subtree $|T_u|$ multiplied by $\delta$: $e_u \le |T_u|\delta$.

The total error in the expected influence is the sum of errors at all nodes:
$$ |\hat{\varsigma} - \varsigma| \le \sum_{u \in V} e_u \le \sum_{u \in V} |T_u|\delta $$
In the worst-case (a line graph), $\sum_{u \in V} |T_u| = \frac{n(n+1)}{2}$. Thus, the total error is $O(n^2 \delta)$. To ensure the total error is $\le \epsilon$, we set $\delta = \frac{\epsilon}{n^2}$, which implies $M = \frac{n^2}{\epsilon}$.
\end{proof}

\begin{lemma}[FPTAS Guarantee]
    With $M = O(n^2/\epsilon)$, \Cref{alg:DPIM_InArb_Final} is a Fully Polynomial-Time Approximation Scheme (FPTAS).
\end{lemma}
\begin{proof}
    As shown in the Error Bound lemma, setting $M = n^2/\epsilon$ ensures the additive error is at most $\epsilon$. 
    Since the number of infected vertices is at least $k$, which is at least $1$, an additive error of at most $\epsilon$ implies a multiplicative error of at most $\epsilon$.
    
    For any node $v$ and budget $k$, the number of unique achievable pairs in $\mathcal{S}(v, k)$ after rounding is at most the number of possible discrete values for $p$ and $\varsigma$. Since $p \in \{0, \delta, \dots, 1\}$ and $\varsigma \in \{0, \delta, \dots, n\}$, the size of the set is:
    $$ |\mathcal{S}(v, k)| \le (M+1)(nM+1) = O(nM^2) $$
    Substituting $M = n^2/\epsilon$, we get $|\mathcal{S}(v, k)| = O(n^5/\epsilon^2)$.
    
    \textbf{Time Complexity:} There are $O(nK)$ states. For each state, we iterate through $O(K)$ budget partitions. In each partition, we compute the Cartesian product of two sets of size $O(nM^2)$, which takes $O(n^2 M^4)$ time. 
    The total time complexity is:
    $$ O(n \cdot K \cdot K \cdot n^2 M^4) = O(n^3 K^2 M^4) = O\left( \frac{n^{11} K^2}{\epsilon^4} \right) $$
    Since the running time is polynomial in both $n$ and $1/\epsilon$, the algorithm is an FPTAS.
\end{proof}

\begin{algorithm}[t]

\caption{DPIM-InArb: Discretized Achievable Pairs DP}

\label{alg:DPIM_InArb_Final}

\SetKwProg{Fn}{Function}{ is}{end}

\SetKwFunction{FCompute}{Compute-S}


\SetKw{KwAnd}{and}

\KwIn{In-arborescence $T=(V,E)$, total budget $K$, discretized set $Q$, edge weights $B$}

\KwOut{Maximum expected influence $\varsigma_{max}$}

\BlankLine

Initialize global hash table $\mathit{Memo} \gets \emptyset$\;

\Fn{\FCompute{$v, k$}}{

    \If{$(v, k) \in \mathit{Memo}$}{

        \Return $\mathit{Memo}[(v, k)]$ \tcp*[r]{Retrieve cached result}

    }

    $S_{v, k} = \emptyset$\;

    \eIf{$v$ is a leaf node}{

        $S_{v, k} \gets S_{v, k} \cup \{(0.0, 0.0)\}$\;

        \If{$k \ge 1$}{

            $S_{v, k} \gets S_{v, k} \cup \{(1.0, 1.0)\}$ \tcp*[r]{Seed at $v$}

        }

    }{

        Let $v_l, v_r$ be children of $v$ with edge weights $b_l, b_r$\;

        \For{$i = 0$ \KwTo $k$}{

            $S_l = \FCompute{$v_l, i$}$\;

            $S_r = \FCompute{$v_r, k-i$}$\;

            \tcp{Case 1: $v$ is not a seed node ($I_{no}$)}

            \ForEach{$(p_l, \varsigma_l) \in S_l$}{

                \ForEach{$(p_r, \varsigma_r) \in S_r$}{
                    $p_{tmp} = 1 - (1 - p_{l} \cdot b_{l})(1 - p_{r} \cdot b_{r})$\;


                    \eIf{v is real}{ $\varsigma_{tmp} = \varsigma_{l} + \varsigma_{r} + p_{tmp}$ }{ $\varsigma_{tmp} = \varsigma_{l} + \varsigma_{r}$ }

                    $S_{v,k} \gets S_{v,k} \cup \{(\text{Round}(p_{tmp}), \text{Round}(\varsigma_{tmp}))\}$


                }

            }

            \tcp{Case 2: $v$ is selected as a seed node ($I_{sel}$)}

            \If{$k \ge 1$ \KwAnd $i \le k-1$ and \text{v is a real node}}{

                $S_l' = \FCompute{$v_l, i$}$\;

                $S_r' = \FCompute{$v_r, k-1-i$}$\;

                \ForEach{$(p_l, \varsigma_l) \in S_l'$}{

                \ForEach{$(p_r, \varsigma_r) \in S_r'$}{
                    $p_{tmp} = 1 - (1 - p_{l} \cdot b_{l})(1 - p_{r} \cdot b_{r})$\;

                    $\varsigma_{tmp} = \varsigma_{l} + \varsigma_{r} + 1.0$\;
                    $S_{v,k} \gets S_{v,k} \cup \{(1.0, \text{Round}(\varsigma_{tmp}))\}$\;

                }

            }

            }

        }

    }

    $\mathit{Memo}[(v, k)] \gets S_{v, k}$\;

    \Return $S_{v, k}$\;

}

\BlankLine

$S_{root} = \FCompute{$root, K$}$\;

\Return $\max \{ \varsigma \mid (p, \varsigma) \in S_{root} \}$\;

\end{algorithm}

%
%
%
\bibliographystyle{splncs04}
\bibliography{ref}

\end{document}